\documentclass[a4paper, 11pt]{article}
\RequirePackage{fullpage}

\usepackage{amsthm,amsmath,amssymb}
\usepackage{mathtools}
\usepackage{subcaption}
\usepackage[dvipsnames]{xcolor}
\usepackage[pdftex,breaklinks,colorlinks,
    citecolor={BlueViolet}, linkcolor={Blue},urlcolor=Maroon]{hyperref}
\usepackage{tikz,pgfkeys}
\usetikzlibrary{quotes}
\usepackage{charter,eulervm}
\usepackage{enumerate}
\usepackage[final,expansion=alltext,protrusion=true]{microtype}

\tikzset{
  filled vertex/.style  = {circle,draw=black,fill=violet,inner sep=1.5pt},
  empty vertex/.style  = {circle, draw, fill = white, inner sep=1.5pt, minimum width=1.5pt}  
}

\theoremstyle{plain} 
\newtheorem{theorem}{Theorem}[section]
\newtheorem{lemma}[theorem]{Lemma}
\newtheorem{corollary}[theorem]{Corollary}
\newtheorem{proposition}[theorem]{Proposition}
\newtheorem{observation}[theorem]{Observation}

\newcommand{\tauplus}{\overrightarrow\tau}  \newcommand{\mssc}{minimum sum set cover}  \newcommand{\msvc}{minimum sum vertex cover}

\title{Minimum Sum Set Cover: Structures and Algorithm}
\author{Zhongyi Zhang\thanks{School of Computer Science and Technology, Shandong University, Qingdao, China. \texttt{zhangzhongyi@mail.sdu.edu.cn}.}
  \and Yixin Cao\thanks{Department of Computing, Hong Kong Polytechnic University, Hong Kong, China.  \texttt{yixin.cao@polyu.edu.hk}.}
 }
\date{}

\begin{document}
\maketitle

\begin{abstract}
A set cover of a hypergraph $H$ is a set of vertices intersecting every hyperedge. In the minimum sum set cover problem, vertices are selected one by one; each edge pays the position of the first vertex that hits it, and the objective is to minimize the total cost.
When $H$ is a graph, this is the minimum sum vertex cover problem.
A solution is specified by a set cover $S$ together with an ordering of its vertices.  While the classical set cover problem seeks to minimize $|S|$, the minimum sum variant favors covering many edges early and may prefer larger covers.  This motivates a natural question: how large can the gap between~$\overrightarrow{\tau}$ and $\tau$ be?

We prove an upper bound $\overrightarrow{\tau} \le \tau \log_{2} \lvert E(H)\rvert$, and show that for any positive~$n$, there exists a hypergraph $H$ on $n + 3$ vertices with $\tau=3$ and $\overrightarrow{\tau}=n$.
For graphs, we obtain stronger bounds: we prove~$\overrightarrow{\tau} \le 2\tau \log_{2} \tau$, improving the bound of Liu et al.\ [Theor. Comput. Sci., 2025], and we construct graphs with~$\overrightarrow{\tau} = \Omega\left( \frac{\tau \log \tau}{\log\log \tau}\right)$, nearly matching this upper bound.

On the algorithmic side, we show that minimum sum set cover is fixed-parameter tractable on bounded-rank hypergraphs, parameterized by~$\overrightarrow{\tau}$, extending the algorithm of Liu et al.\ for graphs (i.e., rank-two hypergraphs).
\end{abstract}

\section{Introduction}

A \emph{hypergraph}~$H$ consists of a vertex set~$V(H)$ and hyperedge set~$E(H)$, where each hyperedge is a nonempty subset of~$V(H)$.
The \emph{rank}~$r(H)$ of a hypergraph~$H$ is the maximum cardinality of any of the edges in~$H$.
A \emph{set cover} of a hypergraph is a set of vertices that intersects every hyperedge.\footnote{In the literature, a set of vertices that covers all hyperedges is typically referred to as a \emph{hitting set}, whereas a set cover is a collection of hyperedges whose union encompasses the entire vertex set. These concepts are equivalent because a set cover of a hypergraph corresponds to a hitting set of its dual hypergraph, where the roles of vertices and edges are interchanged. We follow this established convention, as introduced by Feige et al.~\cite{feige-04-minimum-sum-vertex-cover}.}
There are many variations of this fundamental concept.
A solution consists of an ordered set cover.  When the~$i$th chosen vertex is processed, every previously uncovered hyperedge containing it becomes covered and contributes cost~$i$.
A \emph{\mssc{}} of a hypergraph is a set cover associated with a permutation of its vertices that minimizes the total cost of covering all hyperedges.
When the hypergraph is a graph (rank-two hypergraph), we have \emph{vertex covers} and \emph{\msvc{s}}.

\newcommand{\bset}[1]{{\color{blue} \{#1\}}}
\newcommand{\gset}[1]{{\color{ForestGreen} \{#1\}}}

The classical set cover problem asks for a set cover of minimum cardinality.
In contrast, the minimum-sum variant rewards covering many edges early.
Consequently, a minimum sum set cover does not need to be minimum, or even minimal, as a set cover.
For example, in the hypergraph on vertex set~$\{1, \ldots, 7\}$ and the following eleven hyperedges (blue hyperedges contain vertex~$2$, and green ones contain vertex~$3$)
\[
  \{1, {\color{blue}2}, {\color{ForestGreen}3}\},
  \bset{1, 2, 4}, \bset{1, 2, 5}, \bset{1, 2, 6}, \bset{1, 2, 7},
  \gset{1, 3, 4}, \gset{1, 3, 5}, \gset{1, 3, 6}, \gset{1, 3, 7},
\bset{2, 4, 5}, \gset{3, 6, 7},
\]
the only minimum set cover is $\{2,3\}$, while the only minimum sum set cover is~$\langle 1, 2, 3\rangle$, with cost $9 + 2 + 3 = 14$.
The first vertex~$1$ is not needed to achieve a minimum set cover, but it covers many edges early and is therefore beneficial under the sum objective.

Another peculiarity of minimum sum set cover is that isomorphic components of a (hyper)graph does not need to be treated symmetrically by an optimal solution; see Figure~\ref{fig:disjoint}.

\begin{figure}[ht!]
  \centering
  \begin{tikzpicture}
    \foreach[count=\a from 0] \x/\p in {u/left, v/right} {
      \begin{scope}[xshift={6cm*\a}]
        \def\n{4}
        \node[filled vertex, "$\x_{0}$"] (c) at (\n/2+.5, 2) {};
        \foreach \i in {1, ..., \n} {
          \draw (c) -- (\i, 1) node[filled vertex, "$\x_{\i}$" \p] (v\i) {};
          \foreach \d in {-1, 1} 
          \draw (v\i) -- ++(\d/5, -.7) node[empty vertex] {};
        }
        \foreach \i in {0, \n+1} {
          \draw (c) -- ({\i}, 1) node[empty vertex] {};
        }
      \end{scope}
    }
  \end{tikzpicture}
  \caption{A graph with two identical components.
    All minimum sum vertex covers of this graph are of the form~$\langle u_{0}, v_{1}, v_{2}, v_{3}, v_{4}, u_{1}, u_{2}, u_{3}, u_{4}, v_{0} \rangle$, with the sequence of effective coverages~$\langle 6, 3, 3, 3, 3, 2, 2, 2, 2, 2 \rangle$.  For the left component alone, the only optimal solution is of the form~$\langle u_{1}, u_{2}, u_{3}, u_{4}, u_{0} \rangle$.
  }
  \label{fig:disjoint}
\end{figure}
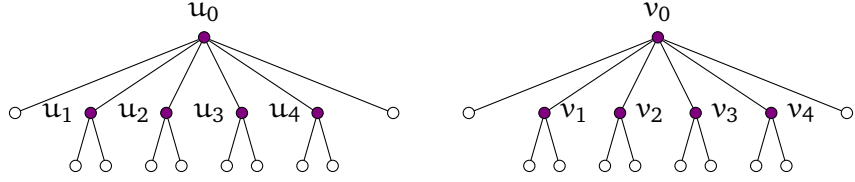

Let $\tau(H)$ denote the size of a minimum (classical) set cover of $H$.
Since a hypergraph may admit multiple optimal minimum sum set covers of varying sizes, we define~$\tauplus(H)$ to be the maximum cardinality of a set cover appearing in an optimal solution.
The example above shows that $\tauplus(H)$ can be strictly larger than $\tau(H)$.
This raises a basic structural question:

\begin{quote}
\emph{What is the possible gap between $\tauplus(H)$ and $\tau(H)$?}
\end{quote}

This question is nontrivial even for graphs, where the ordering objective can force optimal solutions to use substantially more vertices than a minimum vertex cover.
Liu et al.~\cite{cao-25-msvc} showed that for all graphs~$G$, $\tauplus(G) = O(\tau(G)^{1.5}).$
We significantly improve this upper bound.  For hypergraphs, the upper bound has to rely on~$|E(H)|$ as well.

\begin{theorem}\label{thm:mssc-rank}
  \begin{itemize}
  \item  For any graph~$G$ with~$\tau(G) \ge 2$, it holds~$\tauplus(G) \le 2\tau(G) \log_{2} \tau(G)$.
  \item  For any hypergraph~$H$ with~$\tau(H) \ge 2$, it holds~$\tauplus(H) < \tau(H) \log_{2} |E(H)| \le r(H)\tau(H) \log_{2} |V(H)|$.
  \end{itemize}
\end{theorem}

The main observation behind the proof is that at any step, the remaining sub-hypergraph admits a set cover of size at most $\tau(H)$.  Consequently, the vertex maximizing immediate coverage must remove at least a $1/\tau(H)$ fraction of the remaining edges.  This implies that the number of uncovered edges decays exponentially, by a factor of $(1- 1/\tau(H))$ at each step.  
For graphs~$G$, this implies that~$|E(G)|\le \tau(G)^2$ after removing the vertices contained in every minimum vertex cover and every minimum sum vertex cover.

We show by construction that the~$|E(H)|$ term for hypergraphs is inevitable.
Our construction is based on encoding an $n$-vertex graph $G$ into a hypergraph $H_G$ with three ``special'' vertices that form the only minimum set cover, while~$V(G)$ forms the only other minimal set cover.
The minimum sum objective forces almost all vertices in~$V(G)$ to be used before those special vertices.

\begin{theorem}\label{thm:main-1}
  For any integer~$n \ge 3$, there exists a hypergraph~$H$ such that~$\tau(H) = 3$ and~$\tauplus(H) = n$.  \end{theorem}

For graphs, we are only able to furnish an almost tight construction matching the upper bound in Theorem~\ref{thm:main-1}, leaving a gap of~$\log\log \tau(G)$.

\begin{theorem}\label{thm:msvc}
  There are bipartite graphs~$G$ such that~$\tauplus(G) = \Omega\left( \frac{\tau(G) \log \tau(G)}{\log\log \tau(G)}\right)$. 
\end{theorem}

For the lower bound, we use a connected bipartite graph~$B_{q}$ inspired by the seminal construction of Johnson~\cite{johnson-74-approximation} for the degree-greedy algorithm for the vertex cover problem.
We start with a matching of~$2 n^q$ edges on~$L$ and~$R_{0}$.
The set~$L$ is the only minimum vertex cover of the graph.
The other side~$R$ has~$q$ more groups of vertices, $R_{1}, \ldots, R_{q}$, each of~$n^q$ vertices.  They are connected to~$L$ in a way that with the removal of the last group, the graph is decomposed into~$n$ components, each isomorphic to~$B_{q-1}$.

We show that any minimum sum vertex cover proceeds through the graph in a highly constrained way.
In particular, it starts with all the vertices in~$R_{q}$, whose removal leaves $n$ copies of $B_{q-1}$.
For each of them, we either take the vertices in the~$L$ side, or the last group of it (which is a subset of $R_{q-1}$).  
The same pattern must repeat in the copies~$B_{q-2}, \ldots, B_{2}$, leading to a lower bound $\tauplus(B_q) = \Omega(n^{q+1})$.
The claimed bound follows from these by setting~$q=n$.

The approximation and complexity of the minimum sum set cover problem were initiated by Feige et al.~\cite{feige-04-minimum-sum-vertex-cover}, who studied approximation algorithms and hardness for the general problem.
Their reduction implies that it is NP-hard to compute~$\tauplus$.
Recently, Liu et al.~\cite{cao-25-msvc} complemented this by showing that the problem remains NP-hard even when the vertex set of a minimum sum set cover is provided.

We study the parameterized complexity of the minimum sum set cover problem.
Formally, we consider the following decision problem. Given a hypergraph $H$ and integers $k$ and $w$, does there exist a set cover~$S$ of~$H$ and a permutation of~$S$ such that~$|S| \le k$ and the cost is at most~$w$.
Parameterizing by~$w$ alone is not meaningful, since every yes-instance satisfies~$w \ge |E(H)|$.
The more natural parameterization is to use~$\tauplus$.
Liu et al.~\cite{cao-25-msvc} showed that the minimum sum vertex cover problem, minimum sum set cover on rank-two hypergraphs, is fixed-parameter tractable (FPT) with respect to~$\tauplus$.

The algorithm of Liu et al.~\cite{cao-25-msvc} finds the solution set~$S$ by extending a minimal vertex cover~$S'$, which can be found in FPT time.
The most crucial observation behind their low complexity is that for each of the remaining positions of~$S\setminus S'$, one only needs to consider a small number of candidates: since~$S'$ is a vertex cover, the number of edges effectively covered by any vertex in~$S\setminus S'$ depends solely on its position relative to the sequence of~$S'$.
Heavily relying on the fact that each edge involves precisely two vertices, the algorithm cannot be easily adapted to hypergraphs.

We present a simple branching algorithm for the \mssc{} problem on bounded-rank hypergraphs, first using sunflowers to identify a bounded set of candidate vertices, and then exhaustively searching over subsets and permutations of these candidates.
The algorithm iteratively applies the Sunflower Lemma~\cite{erdos-60-sunflower} to the hypergraph induced by the uncovered edges (i.e., $H - S$): as long as this hypergraph is large, we find a sunflower with more than $k - |S|$ petals and branch on the (at most $r(H)$) vertices in its core, since any set cover of size at most $k$ must contain at least one of them. Once no such sunflower exists, the lemma implies that the remaining instance has size bounded as a function of the parameter, and we can exhaustively enumerate all ways to complete the solution set to size at most $k$ and all permutations of this set, selecting one of minimum sum cost.

\begin{theorem}\label{thm:alg}
  The minimum sum set cover problem can be solved in time
  \[
    O\left( r(H)^{r(H)+k} k^{r(H)} |E(H)| + (r(H) k)^{r(H) k} \right).
  \]
\end{theorem}

It remains open whether the problem is fixed-parameter tractable when the rank of the input hypergraph is unbounded.  See more discussion in Section~\ref{sec:conclusion}.

\paragraph{Related work.}

Feige et al.~\cite{feige-04-minimum-sum-vertex-cover} showed that a greedy algorithm approximates \mssc{} within a ratio of 4, and it becomes NP-hard to achieve a constant-approximation better than 4.
On graphs, the approximation ratio can be improved to 16/9~\cite{bansal-23-min-sum-set-cover}, while the lower bound, assuming the Unique Games Conjecture and P $\ne$ NP, is 1.014~\cite{stankovic-25-minimum-sum-vertex-cover}.
It remains to see whether these ratios can be improved if FPT time is allowed.

\section{Upper bounds: Theorem~\ref{thm:mssc-rank}}

We let $H - S$ denote the subhypergraph of $H$ induced by $V(H)\setminus S$, i.e., the hypergraph with vertex set $V(H)\setminus S$ and hyperedge set
\[
\{ e\setminus S\mid e\in E(H), e \cap S = \emptyset \},
\]
and it is further shortened to~$H - v$ when~$S$ consists of a single vertex~$v$.

The \mssc{} problem asks for a set cover~$S$ and a permutation of its vertices.  It is more convenient to extend this permutation into a bijection~$\sigma: V(H) \rightarrow \{1, 2, \ldots, |V(H)|\}$ by arbitrarily assigning vertices in~$V(H)\setminus S$ into $\{|S| + 1, \ldots, |V(H)|\}$.
Note that any ordering of $V(H)$ induces a feasible extension of an ordering of a set cover.
For an ordering~$\sigma: V(H) \rightarrow \{1, 2, \ldots, |V(H)|\}$, define the \emph{effective coverage} of the $i$th vertex as
\[
  r_{\sigma}(i) = \left| \{e\in E(H) \mid \min_{v\in e} \sigma(v) = i \} \right|.
\]
Abusing notation, we also write~$r_{\sigma}(\sigma(v))$ as~$r_{\sigma}(v)$.
The cost of \mssc{s} of~$H$, denoted as~$\phi(H)$, is
\[
  \phi(H) = \min_{\sigma} \sum_{i=1}^{|V(H)|} i\cdot r_{\sigma}(i) = \min_{\sigma} \sum_{e \in E(H)} \min_{v\in e} \sigma(v).
\]

Throughout the paper, we will consider solutions~$(S, \sigma)$ with~$r_{\sigma}(i) > 0$ for all~$i = 1, \dots, |S|$.
Indeed, if~$r_{\sigma}(|S|) = 0$, then the last vertex of $S$ is redundant; if~$r_{\sigma}(i) = 0$ for some~$i < |S|$, then the vertex at position $i$ does not cover any edge at that position; removing it and shifting all later vertices one step earlier yields another solution with strictly smaller cost.
Under this assumption, the set~$S$ is implicit from the ordering~$\sigma$: its vertices are exactly those that appear before any edge is first covered.  Equivalently,
\[
  |S| = \max_{e \in E(H)} \min_{v\in e} \sigma(v).
\]
Henceforth, we will simply use~$\sigma$ to denote a solution, and~$S_\sigma$ the implied vertex set.

\begin{observation}\label{lem:remaining}
  Let~$H$ be a hypergraph on~$n$ vertices, and~$\langle v_{1}, v_{2}, \ldots, v_{n}\rangle$ a \mssc{} of~$H$.  For all~$i = 1, \ldots, n - 1$, the ordering~$\langle v_{i}, v_{i + 1}, \ldots, v_{n}\rangle$ is a \mssc{} of~$H - \{v_{1}, v_{2}, \ldots, v_{i - 1}\}$.
\end{observation}
\begin{proof}
  If~$H - \{v_{1}, v_{2}, \ldots, v_{i - 1}\}$ admits a better solution~$\sigma$, then concatenating~$\langle v_{1}, v_{2}, \ldots, v_{i-1}\rangle$ and~$\sigma$ yields a solution of strictly smaller total cost.
\end{proof}

In an optimal solution, the numbers~$r_{\sigma}(\cdot)$ are non-increasing.
\begin{lemma}\label{lem:fix position}
  In any optimal solution~$\sigma$ of~$H$, it holds~
  \[
    r_{\sigma}(1)\ge r_{\sigma}(2)\ge \cdots \ge r_{\sigma}(|V(H)|).
  \]
\end{lemma}
\begin{proof}
  By Observation~\ref{lem:remaining}, it suffices to show that~$r_{\sigma}(1)\ge r_{\sigma}(2)$.
  Let~$u$ and~$v$ be the vertices with~$\sigma(u) = 1$ and~$\sigma(v) = 2$.
Let~$\sigma'$ be obtained from~$\sigma$ by switching the order of~$u$ and~$v$.
  For each hyperedge containing~$u$ but not~$v$, the cost is~$1$ in~$\sigma$ and~$2$ in~$\sigma'$.
  For each hyperedge containing~$v$ but not~$u$, the cost is~$2$ in~$\sigma$ and~$1$ in~$\sigma'$. 
  All the other hyperedges have the same cost in~$\sigma$ and~$\sigma'$.
  By the optimality of~$\sigma$, we have~$r_{\sigma}(1) - r_{\sigma}(2) \ge 0$.
\end{proof}

Lemma~\ref{lem:fix position} implies a lower bound for the number of hyperedges the first vertex needs to cover.
\begin{lemma}\label{lem:initial}
  In any \mssc{}~$\sigma$, it holds~$r_{\sigma}(1) \ge |E(H)|/\tau(H)$.
\end{lemma}
\begin{proof}
  Let~$S$ be a minimum set cover of~$H$, hence~$|S| = \tau(H)$.
  We take an ordering~$\pi$ of~$V(H)$ such that~$\pi(x) \le \tau(H)$ if and only if~$x\in S$, and
\[
  r_{\pi}(1) \ge r_{\pi}(2) \ge \cdots \ge r_{\pi}(\tau(H)).
\]
Such an ordering can be produced as follows.
Let~$H_{1} = H$.
For~$i = 1, \ldots, \tau(H)$, we take a vertex~$v_{i}$ in~$S$ with the maximum degree in~$H_{i}$, and let~$H_{i+1} = H_{i} - v_{i}$.
The ordering is~$\langle v_{1}, \ldots, v_{\tau(H)} \rangle$ followed by other vertices in an arbitrary order.
For~$i = 1, \ldots, \tau(H) - 1$,
\[
  r_{\pi}(i+1) = d_{H_{i+1}}(v_{i+1}) \le d_{H_{i}}(v_{i+1}) \le d_{H_{i}}(v_{i}) = r_{\pi}(i).
\]
The maximum possible cost of~$\pi$ is when~$r_\pi(i)$ is either~$\left\lceil {|E(H)| \over \tau(H)} \right\rceil$ or~$\left\lfloor {|E(H)| \over \tau(H)} \right\rfloor$ for all~$i=1, \ldots, \tau(H)$.
If~$r_{\sigma}(1) < |E(H)|/\tau(H)$, then~$r_\sigma(i) < |E(H)|/\tau(H)$ for all~$i$.
By Proposition~\ref{lem:bettersort}, the cost of~$\sigma$ is higher than~$\pi$, a contradiction.
\end{proof}

We are now ready to prove Theorem~\ref{thm:mssc-rank}.

\begin{proof}[Proof of Theorem~\ref{thm:mssc-rank}]
  We deal with the hypergraph first.
  Let~$\tau = \tau(H)$ and~$\tauplus = \tauplus(H)$ and~$m = |E(H)|$.
  For~$i = 1, \ldots, n$, let~$H_{i} = H - \{v_{1}, \ldots, v_{i-1}\}$ and~$m_{i} = |E(H_{i})|$.
  Note that~$m_{1} = m$ and~$m_{i+1} = m_{i} - r_\sigma(i)$ for~$i = 1, \ldots, n - 1$.
  By Observation~\ref{lem:remaining} and Lemma~\ref{lem:initial},
  \[
    r_\sigma(i) \geq \frac{m_{i}}{\tau (H_{i})} \geq \frac{m_{i}}{\tau}.
  \]
  Thus,
  \[
    m_{i+1} = m_{i} - r_\sigma(i) \leq m_{i} - \frac{m_{i}}{\tau} = \frac{\tau-1}{\tau}m_{i},
  \]
  which implies
  \[
    m_{\tauplus} \le \left(\frac{\tau-1}{\tau} \right)^{\tauplus - 1} m_{1} = \left(\frac{\tau-1}{\tau} \right)^{\tauplus - 1} m.
  \]
Since~$\tau \ge 2$ and~$m_{\tauplus} \ge 1$, we have
  \[
    m \ge \left(1+ {1\over \tau - 1} \right)^{(\tau - 1) \left(\frac{\tauplus-1}{\tau - 1}\right)}\ge 2^{\frac{\tauplus-1}{\tau - 1}}, 
  \]
  and
  \begin{equation}
    \label{eq:8}
    \tauplus \le (\tau - 1) \log_{2} m + 1 < \tau \log_{2} m.
  \end{equation}

Now consider a graph~$G$.
  We fix a minimum vertex cover~$S^*$, and a \msvc{}~$\sigma$ of~$G$.
  Let~$v$ be the first vertex in~$\sigma$.
  We may assume without loss of generality that~$v\not\in S^*$, as otherwise we consider the subgraph~$G - v$ (Observation~\ref{lem:remaining}).
  Repeating this argument while the first chosen vertex lies in~$S^*$, we eventually arrive at a subgraph in which the first chosen vertex is outside a fixed minimum vertex cover, without changing the validity of the bound.
  As a result,~$N(v)\subseteq S^*$, and hence~$r_{\sigma}(1) = d(v) \le |S^*| = \tau.$
  By Lemma~\ref{lem:initial},~$m \le \tau^{2}$.  
  The statement then follows from~\eqref{eq:8}.
\end{proof}

The proof of Lemma~\ref{lem:fix position} is simple because switching two juxtaposed vertices has no impact on the effective coverages of other vertices.  In the later proofs, we need to switch vertices that are not juxtaposed, for which we need a simple arithmetic inequality for comparing two solutions of the same hypergraph.

\begin{proposition}\label{lem:bettersort}
  Let~$\langle s_{1}, \ldots, s_{n}\rangle$ and~$\langle s'_{1}, \ldots, s'_{n}\rangle$ be two non-increasing sequences of nonnegative integers such that~$\sum_{i=1}^{n} s_{i} = \sum_{i=1}^{n} s'_{i}$.  If there exists an integer~$t \in \{1,\dots,n\}$ such that
  \begin{align}
    \label{eq:2}
    s_i &\le s'_i, \quad i = 1, \ldots, t,
    \\
    \label{eq:4}
    s_i &\ge s'_i, \quad i = t+1, \ldots, n.
  \end{align}
then
  \[
    \sum_{i=1}^{n} i\cdot s_{i} \ge \sum_{i=1}^{n} i\cdot s'_{i},
  \]
  with equality if and only if all inequalities in \eqref{eq:2} and \eqref{eq:4} are equalities.
\end{proposition}
\begin{proof}
  For $i=1,\dots,n$, define $d_i = s_i - s'_i$, and~$D_k = \sum_{i=1}^k d_i$.
Then
  \[
    \sum_{i=1}^n d_i \;=\; \sum_{i=1}^n s'_i - \sum_{i=1}^n s_i \;=\; 0.
  \]
  From \eqref{eq:2} and \eqref{eq:4} we have~$d_i \le 0$ for $i \le t$, and~$d_i \ge 0$ for $i > t$.
  Then $D_n = 0$, and~$D_i \le 0$ for all $i = 1,\dots,n$: when~$i \le t$, $D_i$ is the sum of non-positive terms, and for $i > t$, it is the negation of the sum of nonnegative terms.

Thus,
\[
\sum_{i=1}^n i\,(s_i - s'_i)
= \sum_{i=1}^n i\,d_i
= \sum_{i=1}^n \sum_{k=1}^i d_i
= \sum_{k=1}^n \sum_{i=k}^n d_i
= \sum_{k=1}^n (D_n - D_{k-1})
= -\sum_{k=1}^{n-1} D_k \ge 0,
\]
where we use the convention $D_0 = 0$.

Moreover, equality holds if and only if $D_k = 0$ for all $k = 1,\dots,n-1$, which in turn is equivalent to $d_i = 0$ for all $i$, i.e.\ $s_i = s'_i$ for all $i$. 
\end{proof}

We will frequently switch two vertices in a solution to obtain another solution.  The following corollary will help us compare their costs.

\begin{corollary}\label{cor:switching}
  Let $\sigma$ be an ordering of $V(H)$ and let $\sigma'$ be obtained from $\sigma$ by swapping two vertices $u$ and $v$ with $\sigma(u) < \sigma(v)$.  Denote $p = \sigma(u)$ and $q = \sigma(v)$.  If
\begin{enumerate}
  \item $r_{\sigma'}(p) > r_\sigma(p)$, and
  \item $r_{\sigma'}(i) \le r_\sigma(i)$ for all $i$ with $p < i \le q$.
\end{enumerate}
Then $\sigma'$ has strictly smaller total cost than $\sigma$.
\end{corollary}
\begin{proof}
Let
\[
s_i = r_\sigma(i), \qquad s'_i = r_{\sigma'}(i), \qquad i = 1,\dots,n.
\]
Clearly $\sum_{i=1}^n s_i = \sum_{i=1}^n s'_i = |E(H)|$, since both orderings cover
all edges exactly once.

By construction, $\sigma$ and $\sigma'$ coincide on all positions $i < p$, hence
$s_i = s'_i$ for $i < p$. At position $p$ we have $s'_{p} > s_{p}$ by the first assumption.  For positions $p < i \le q$, the second assumption gives
$s'_i \le s_i$. For $i > q$, $\sigma$ and $\sigma'$ again coincide, hence
$s_i = s'_i$ for all $i > q$.
The statement follows from Proposition~\ref{lem:bettersort} with~$t = p$.
\end{proof}

\section{The lower bound: Theorem~\ref{thm:main-1}}

Consider a hypergraph on vertex set~$A\cup B$ with~$|A| = n$ and~$|B| = 3$.
We will construct hyperedges such that a solution with~$A$ preceding~$B$ has a cost smaller than a solution with~$A$ following~$B$, while all other solutions are even worse than the latter.
This construction establishes Theorem~\ref{thm:main-1}.

Given a graph~$G$ on~$n$ vertices~$\{v_{1}, \ldots, v_{n}\}$ and~$m$ edges, we define a hypergraph~$H_{G}$ as follows.
The vertex set is~$A\cup B$, where~$A = V(G)$ and~$B = \{b_{1}, b_{2}, b_{3}\}$.
The hyperedges are constructed based on the subsets of~$A$.
For each nonempty subset~$X$ of~$A$ that is not in~$\binom{A}{2}\setminus E(G)$ (i.e., $|X|\ne 2$ or~$X\in E(G)$), introduce three hyperedges~$X\cup \{b_{1}\}$,~$X\cup \{b_{2}\}$, and~$X\cup \{b_{3}\}$.
Formally, the edge set of~$H_{G}$ is:
\begin{align*}
  E(H_G) = & \big\{\{v_{i}, b_{k}\} \mid i \in \{1, \ldots, n\}, k \in \{1, 2, 3\}\big\}
  \\
  \cup & \big\{\{v_{i}, v_{j}, b_{k}\} \mid v_{i} v_{j}\in E(G), k \in \{1, 2, 3\}\big\}
  \\
  \cup & \big\{ X\cup \{b_{k}\} \mid X \subseteq A, |X| \ge 3, k \in \{1, 2, 3\} \big\}.
\end{align*}

\begin{proposition}\label{lem:construction-msc}
  Both~$A$ and~$B$ are minimal set covers of~$H_{G}$, and they are the only minimal set covers.
\end{proposition}
\begin{proof}
  Let~$S$ be a set cover of~$H_{G}$.
  If~$A\not\subseteq S$, there exists a vertex~$a\in A\setminus S$.
  Since~$S$ must cover the hyperedges~$\{a, b_{1}\}$,~$\{a, b_{2}\}$, and~$\{a, b_{3}\}$, and~$a \notin S$, it must contain all three vertices in~$B$.
  Thus, if~$S$ does not contain~$A$, it must contain~$B$.
  Since~$A$ and~$B$ are disjoint and cover all edges, they are the only minimal set covers.
\end{proof}

We now analyze the cost of a solution~$\pi_{B}$ that orders~$B$ before~$A$, where the order of~$b_{1}, b_{2}, b_{3}$ is irrelevant.
For~$k = 1, 2, 3$,
\begin{equation}
  \label{eq:5}
  r_{\pi_B}(k) = n + m + \left(2^n - 1 - n - \binom{n}{2}\right) = 2^{n} - 1 - \binom{n}{2} + m.
\end{equation}
Since~$r_{\pi_B}(k) = 0$ for~$k \ge 4$, the cost is
\begin{equation}
  \label{eq:6}
  \sum_{i=1}^{n+3} i\cdot r_{\pi_B}(i) = \sum_{i=1}^{3} i\cdot r_{\pi_B}(i) = 6 \left(2^{n} - 1 - \binom{n}{2} + m\right).  
\end{equation}

In the following two propositions, we show that any solution that orders~$A$ before~$B$ is strictly better than~$\pi_B$, while any other solution is strictly worse than~$\pi_B$.

\begin{proposition}\label{lem:a-to-b}
  The cost of any solution that orders~$A$ before~$B$ is smaller than that of~$\pi_{B}$.
\end{proposition}
\begin{proof}
  Let~$\pi$ be a solution for~$H_G$ with~$\pi(b_k)=n+k$ for all~$k=1,2,3$, and let~$\sigma$ be the subsequence of the first~$n$ vertices in~$\pi$.  Then~$\sigma$ is an ordering of~$A=V(G)$, and hence a solution of~$G$.

  For each~$i=1,\dots,n$, define
  \[
    \delta_i = (n-i)-r_\sigma(i).
  \]
  Since the~$i$th vertex in~$\sigma$ can cover at most the~$n-i$ remaining vertices of~$G$, we have~$\delta_i\ge 0$.  Moreover,
\[
  \sum_{i=1}^{n} \delta_{i} = \sum_{i=1}^n (n-i) - \sum_{i=1}^{n} r_{\sigma}(i) = \binom{n}{2} - m.
\]

We now compute the effective coverage of the~$i$th vertex in~$\pi$.  After the first~$i-1$ vertices of~$A$ have been selected, there remain exactly~$n-i$ vertices of~$A$ not yet chosen.  The~$i$th selected vertex covers:
\begin{itemize}
\item exactly~$3r_\sigma(i)$ hyperedges of the form~$\{v_x,v_y,b_k\}$ corresponding to edges~$v_xv_y\in E(G)$ first hit at position~$i$ in~$\sigma$; and
\item exactly~$3\bigl(2^{n-i}-(n-i)\bigr)$ hyperedges of the form~$X\cup\{b_k\}$ with either~$|X|=1$ or~$|X|\ge 3$, since among the~$2^{n-i}$ subsets of the remaining~$n-i$ vertices, exactly~$n-i$ are singletons.
\end{itemize}
Hence,
\[
  r_{\pi}(i) = 3 r_{\sigma}(i) + 3 \left( 2^{n - i} - (n - i) \right) = 3 \left( 2^{n - i} - \delta_{i} \right).
\]
The cost of~$\pi$ is thus
\begin{align*}
\sum_{i=1}^{n+3} i\cdot r_{\pi}(i) = &\sum_{i=1}^{n} i\cdot r_{\pi}(i)
  \\
 =&\sum_{i = 1}^n 3 i \left( 2^{n - i} - \delta_{i} \right)
\\ 
  = & \left( 3 \sum_{i = 1}^n i \cdot 2^{n-i} \right) - \left( 3 \sum_{i = 1}^n i \delta_{i} \right).
\end{align*}
  Since~$i\ge 2$, we have
  \[
    \sum_{i=1}^n i\delta_i
    = \delta_1 + \sum_{i=2}^n i\delta_i
    \ge \delta_1 + 2\sum_{i=2}^n \delta_i
    = 2\sum_{i=1}^n \delta_i - \delta_1.
  \]
  Hence,
  \begin{align*}
    \sum_{i=1}^{n+3} i\cdot r_\pi(i)
    & \le \left( 3 \sum_{i = 1}^n i \cdot 2^{n-i} \right) - \left( 6 \sum_{i = 1}^n \delta_{i} \right) + 3 \delta_{1} 
    \\
    &= 3\bigl(2^{n+1}-(n+2)\bigr) - 6\left(\binom{n}{2}-m\right) + 3\delta_1 \\
    &= 6 \cdot 2^{n} - 3 (n + 2) - 6 \binom{n}{2} + 6 m + 3 \delta_{1}.
    \\
    &< 6 \cdot 2^{n} - 6 - 6 \binom{n}{2} + 6 m,
  \end{align*}
  where the last inequality holds because $ \delta_1 = (n-1)-r_\sigma(1) \le n-1 < n$.
\end{proof}

\begin{lemma}\label{lem:construction-mssc}
  For any graph~$G$, every minimum sum set cover of~$H_{G}$ places all vertices of~$B$ after all vertices of~$A$.  
\end{lemma}
\begin{proof}
Let~$\pi$ be a solution in which~$\{\pi(b_{1}), \pi(b_{2}), \pi(b_{3})\}$ is neither~$\{1, 2, 3\}$ nor~$\{n+1, n+2, n+3\}$.
  Since the vertices in~$B$ are symmetric, we may assume without loss of generality that
  \[
    \pi(b_{1}) < \pi(b_{2}) < \pi(b_{3}).
  \]
  We use induction on~$|V(G)|$ to show~$\pi$ is not optimal.
  The base case is trivial when~$A$ consists of a single vertex.
    When $n=2$, a direct case analysis on whether~$v_{1} v_{2}\in E(G)$ shows that every ordering with a vertex of~$B$ among the first two positions has larger cost than an ordering with~$A$ first.  

  Consider the general case~$n > 2$.
  If~$\pi(b_{1}) \ne 1$, let~$v$ be the first vertex in the ordering~$\pi$ and let~$\pi'$ be the ordering obtained from~$\pi$ by removing~$v$ and shifting indices.
  By Observation~\ref{lem:remaining},~$\pi'$ is a \mssc{} of~$H_{G} - v$.
  Since~$H_{G} - v$ is isomorphic to~$H_{G - v}$, this violates the induction hypothesis.

  Hence~$\pi(b_{1}) = 1$, and $r_{\pi}(1) = r_{\pi_{B}}(1) = 2^{n} - 1 - \binom{n}{2} + m$.
By Lemma~\ref{lem:fix position}, for~$i = 2, 3$,
  \[
    r_{\pi}(i) \le r_{\pi}(1) = r_{\pi_{B}}(i).
  \]
  Moreover,~$r_{\pi}(i) \ge 0 = r_{\pi_{B}}(i)$ when $i \ge 4$, where the inequality is strict for~$i = 4$ by Proposition~\ref{lem:construction-msc}.
  By Proposition~\ref{lem:bettersort} with~$t = 3$, we can conclude that the cost of~$\pi$ is strictly larger than~\eqref{eq:6}.  

  Thus,~$\pi$ is not an optimal solution.  Combining with Proposition~\ref{lem:a-to-b}, we can conclude that any optimal solution must put~$B$ at the end.
\end{proof}

Theorem~\ref{thm:main-1} now follows immediately from the preceding discussion.
\begin{proof}[Proof of Theorem~\ref{thm:main-1}]
  By Proposition~\ref{lem:construction-msc},~$\tau(H_{G}) = |B| = 3$.
  By Lemma~\ref{lem:construction-mssc}, any minimum sum set cover must order all vertices of~$A$ before~$B$. Thus, the set cover associated with the optimal permutation is~$A$, implying~$\tauplus(H_{G}) = |A| = n$.
\end{proof}

\section{The lower bound for graphs: Theorem~\ref{thm:msvc}}

Fix a positive integer~$n$.  For a positive integer $q \le n$, let $B_q$ be the bipartite graph with bipartition $L \uplus R$, where $|L| = 2n^{q}$ and $|R| = (q+2)n^{q}$. Partition $R$ as
\[
R = R_0 \uplus R_1 \uplus \cdots \uplus R_q
\]
with $|R_0| = 2n^{q}$ and $|R_i| = n^{q}$ for each $i = 1, \ldots, q$.

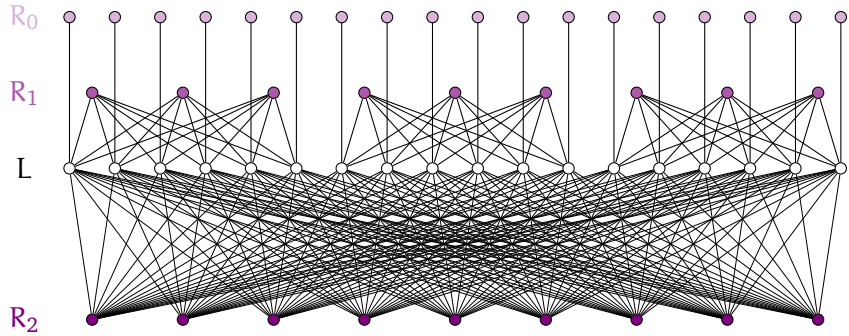
\begin{figure}[ht!]
  \centering
  \begin{tikzpicture}[xscale=1.2]
    \def\n{3}
    \foreach \i in {1, ..., \inteval{\n*\n}}
\node[filled vertex,fill=violet] (r2\i) at (\i-.25, 0) {};

    \foreach[count=\a from 0] \i in {1, ..., \n} {
      \begin{scope}[xshift={\n cm*\a}]
\foreach \j in {1, ..., \n} 
        \node[filled vertex,fill=violet!65] (r1\i\j) at (\j-.25, 3) {};
        
        \foreach \j in {1, ..., \inteval{\n*2}} {
          \node[empty vertex] (l\i\j) at (\j/2, 2) {};

          \draw[thin] (l\i\j) -- ++(0, 2) node[filled vertex,fill=violet!30] {};

          \foreach \k in {1, ..., \n} 
          \draw[thin] (l\i\j) -- (r1\i\k);
          \foreach \k in {1, ..., \inteval{\n*\n}}
          \draw[ultra thin] (l\i\j) -- (r2\k);
        }
      \end{scope}
    }
\node at (0, 2) {$L$};
    \node[violet!30] at (0, 4) {$R_{0}$};
    \node[violet!65] at (0, 3) {$R_{1}$};
    \node[violet] at (0, 0) {$R_{2}$};
  \end{tikzpicture}
  \caption{The graph~$B_{2}$ with~$n = 3$.
  }
  \label{fig:construction}
\end{figure}

Define the edges as follows. First, add all edges between $L$ and $R_q$, and set $\mathcal{L}_q = \{L\}$.
For each $i = q-1, \ldots, 1$, obtain $\mathcal{L}_i$ from $\mathcal{L}_{i+1}$ by evenly refining every part into $n$ subparts; thus $\mathcal{L}_i$ is a partition of $L$ into $n^{q-i}$ parts, each of size $2n^{i}$. Independently, partition $R_i$ into $n^{q-i}$ parts of equal size, and for each $j$, make the bipartite subgraph between the $j$th part of $\mathcal{L}_i$ and the $j$th part of this partition of $R_i$ complete.
Thus every vertex in a part of $\mathcal{L}_i$ has the same neighbors in $R_i$.
Finally, index the vertices of $L$ and $R_0$ as $\{1,\dots,2n^{q}\}$, and for each $i = 1,\dots,2n^{q}$, add an edge between the $i$th vertex of $L$ and the $i$th vertex of $R_0$.
See Figure~\ref{fig:construction} for an illustration.

In summary, $|V(B_{q})| = (q + 4) n^{q}$, $|E(B_{q})| = 2 n^{q} \frac{n^{q+1} - 1}{n-1} $, and the degree of a vertex~$v\in V(B_{q})$ is
\[
  d(v) =
  \begin{cases}
    \frac{n^{q+1} - 1}{n-1}  & v\in L,
    \\
    1 & v\in R_0,
    \\
    2 n^i  & v\in R_{i}, i = 1, \ldots, q.
  \end{cases}
\]
The set~$L$ is the only minimum vertex cover of~$B_{q}$: the subgraph induced by~$L\cup R_{0}$ is an induced matching.
\begin{proposition}\label{lem:vertex-cover}
  $\tau(B_{q}) = 2 n^{q}$.  The cardinality of a vertex cover of~$B_{q}$ that does not contain~$L$ is at least~$3 n^{q}$.
\end{proposition}
\begin{proof}
  The edges between $L$ and $R_0$ form an induced matching of size $2n^q$. Any vertex cover must contain at least one endpoint of each such edge, so $\tau(B_q) \ge 2n^q$. Taking $L$ itself yields a vertex cover of size $2n^q$, so $\tau(B_q) = 2n^q$.

  If a vertex cover $C$ does not contain all vertices of $L$, then there exists $u\in L\setminus C$. The edge between $u$ and its unique neighbor in $R_0$ forces that neighbor into $C$. Repeating this argument shows that $C$ must contain all of $R_0$ and at least $n^q$ vertices of $L$, so $|C| \ge 3n^q$.
\end{proof}

We will show that~$\tauplus(B_{n}) = \Omega(n^{n+1})$.
In the rest, we assume that~$n \ge 7$.

Since~$N(v_{2}) \subsetneq L = N(v_{1})$ for any~$v_{1}\in R_{q}$ and~$v_{2}\in R\setminus R_{q}$, a minimum sum vertex cover of~$B_{q}$ will not choose any vertex in~$R\setminus R_{q}$ before~$R_{q}$.  In other words, before the first vertex in~$R\setminus R_{q}$, at least one of~$L$ and~$R_{q}$ has been exhausted.  As we will see, a minimum sum vertex cover will start with~$R_{q}$.
After that, the remaining graph~$B_{q} - R_{q}$ is decomposed into~$n$ components, each is a copy of~$B_{q-1}$.  For a positive integer~$p$ and a graph~$G$, we use~$p G$ to denote the disjoint union of~$p$ copies of~$G$.

\begin{observation}\label{lem:subgraph}
  The subgraph~$B_{q} - R_{q}$ is isomorphic to~$n B_{q-1}$.
\end{observation}
As we have seen in Figure~\ref{fig:disjoint}, these copies may behave differently in a minimum sum vertex cover.  We show that in each of the copies, 
\begin{enumerate}
\item the vertices in~$L$ are selected first (which cover all edges);
\item all vertices of~$R_{q-1}$ are selected first, which reduces the remaining copy into~$n$ copies of~$B_{q-2}$; or
\item vertices of a subset of~$L\cup R_{q-1}$ are selected first. 
\end{enumerate}
The main technical lemma is that the third case will never happen in a minimum sum vertex cover.
Let~$G_{p, q} = p B_{q}$.
For~$i = 1, \ldots, p$, we use~$B^{i}_{q}$ to refer to the~$i$th copy of~$B_{q}$ in~$G_{p, q}$, with vertex set~$L^{i}$,~$R_{0}^{i}$, $\ldots$,~$R_{q}^{i}$.  For a solution~$\sigma$ of~$G_{p, q}$, let~$\sigma^{i} = \sigma|_{V(B^{i}_{q})}$, the sub-ordering of~$\sigma$ induced by~$V(B^{i}_{q})$.

We next prove that in each copy, an optimal solution cannot mix a proper prefix of~$R_{q}$ with later vertices in any beneficial way. This forces each copy to begin either with all of~$R_{q}$ or with all of~$L$.

\begin{lemma}\label{lem:main}
  Let~$\sigma$ be an optimal solution for~$G_{p, q}$.  For all~$i = 1, \ldots, p$, the subsequence~$\sigma^{i}$ starts with either~$R_q^i$ or~$L^i$.
\end{lemma}

Before presenting the proof of Lemma~\ref{lem:main}, we use it to prove Theorem~\ref{thm:msvc}.
We start with bounds on the cost of optimal solutions of~$G_{p, q}$, denoted $\phi(G_{p, q})$.  While the lower bound is nearly trivial, we will upper-bound by explicitly constructing a solution and computing its cost.
Note that the number of edges in $G_{p, q}$ is
\begin{equation}
  \label{eq:1}
 |E(G_{p, q})| = 2 p n^q \sum_{i=0}^{q} n^i = 2 p n^q \frac{n^{q + 1} - 1}{n - 1}.
\end{equation}

\begin{lemma}\label{eq:msvc_cost_bound}
  $p^2 n^{3 q - 1} (n + 2) < \phi(G_{p, q}) < p^2 n^{3 q - 1} (n + 6)$.
\end{lemma}
\begin{proof}
  For the lower bound, suppose that~$\sigma$ is an optimal solution of~$G_{p, q}$.  Since the maximum degree is~$2 n^{q}$, we have~$r_{\sigma}(i) \le 2 n^{q}$ for all~$i = 1, \ldots, |V(G_{p, q})|$.
  Let~$t = p \frac{n^{q + 1} - 1}{n - 1}$.
  Among all non-increasing sequences summing to~$|E(G_{p, q})|$ and bounded above by~$2 n^{q}$, the sequence that starts with as many terms equal to~$2 n^{q}$ as possible minimizes the weighted sum by Proposition~\ref{lem:bettersort}.  Thus,
  \[
    \phi(G_{p, q}) >
    \quad \sum_{i = 1}^{\mathclap{|V(G_{p, q})|}} i \cdot s'_{i} = 
    \sum_{i = 1}^{t} i \cdot 2 n^q = t(t+1) n^q
    > t^2 n^q > p^2 (n^q + n^{q-1})^2 n^q > p^2 n^{3 q - 1} (n + 2),
  \]
  where the third inequality holds because~$\frac{n^{q+1} - 1}{n - 1} > n^q + n^{q-1}$.
  
  For the upper bound, we use an induction on~$q$.
  In the base case,~$q = 1$.  Note that~$L^{1}\cup L^{2}\cup \cdots\cup L^{p}$ is a vertex cover of~$G_{p, q}$.  The cost of a solution starting with~$\langle R_{1}^{1}, R_{1}^{2}, \ldots, R_{1}^{p}, L^{1}, L^{2}, \ldots, L^{p} \rangle$ is
\[
    \sum_{i=1}^{p n} i \cdot 2n + \sum_{i=p n+1}^{3p n} i \cdot 1
    = (p^2n^3 + p n^2) + (4 p^2n^2 + p n)
    < p^2 n^2 (n+6).
  \]
  
  In the general case, $q \geq 2$.  We consider the solution~$\sigma$ starting with~$R_{q}^{1}, R_{q}^{2}, \ldots, R_{q}^{p}$, followed by an optimal solution for~$ G_{p n, q - 1}$.  Its cost is
  \begin{align*}
    \sum_{i=1}^{4 p n^{q}} i\cdot r_{\sigma}(i) =& \sum_{i=1}^{p n^{q}} i \cdot 2 n^{q} + \sum_{i=p n^{q}+1}^{4 p n^{q}} i \cdot r_{\sigma}(i)
    \\
    =& \frac{p n^q(p n^q+1)}{2} 2 n^q +
         		\phi \left( G_{p n, q-1} \right) + p n^q  | E \left( G_{p n, q-1} \right) |
    \\
    <& p^2 n^{3q} + p n^{2 q} + p^2n^{3q-2}(n+6) + 2p^2n^{3q-1} +3p^2n^{3q-2}
    \\
    =& p^2 n^{3q} + p n^{2 q} + 3p^2n^{3q-1} + 9p^2n^{3q-2}
    \\
    <& p^2 n^{3 q-1} (n+6),
  \end{align*}
  where the first inequality follows from the induction hypothesis, and the term~$p n^q  | E \left( G_{p n, q-1} \right) |$ accounts for the shift in positions when appending the optimal solution of $G_{p n, q-1}$ after the first $p n^q$ vertices.
  This concludes the proof.
\end{proof}	

\begin{lemma} \label{eq:msvc_multiparts}
  For any~$p \geq n$ and~$q > 1$, it holds $\tauplus(G_{p, q}) > p n^q + \tauplus\left( G_{p (n - 5), q - 1} \right)$.
\end{lemma}
\begin{proof}
  Let~$\sigma$ be an optimal solution.
  By Lemma~\ref{lem:main}, there exists a number~$a\in \{0, \ldots, p\}$ such that~$\sigma$ starts with
  \begin{itemize}
  \item all vertices in $R_{q}^1 \cup R_{q}^2 \cup \cdots \cup R_{q}^a$, with effective coverage~$2 n^{q}$; and
  \item all vertices in $L^{a+1}, L^{a+2}, \dots, L^{p}$, with effective coverage~$\frac{n^{q+1} - 1}{n - 1}$,
  \end{itemize}
  followed by an optimal solution of~$G_{a n, q - 1}$.

  We split the cost of~$\sigma$ into two parts accordingly.
  The cost of the first~$(2 p - a)n^q$ vertices is
  \begin{align*}
    & \sum_{i=1}^{a n^q} i \cdot 2n^q + \sum_{i = a n^q + 1}^{(2 p - a)n^q} i \cdot \frac{n^{q+1} - 1}{n - 1}
    \\
    = & \frac{a n^q ( a n^q+1)}{2} \cdot 2 n^q + \frac{(2p-2a)n^q(2 p n^q+1)}{2} \frac{n^{q+1} - 1}{n - 1}
    \\
    > &  a^2n^{3q} + {(p - a)n^q(2 p n^q)} (n^q + n^{q-1})
    \\
    = & a^2n^{3q} + 2 p^2 n^{3 q} - 2 a p n^{3 q} + 2 p^2 n^{3 q -1} - 2 a p n^{3 q-1}.
  \end{align*}
  The suffix is an optimal solution of~$G_{a n, q-1}$, together with an additional cost~$(2 p - a) n^q | E(G_{a n, q-1}) |$ because every one of its covered edges is delayed by~$(2 p - a) n^q$ positions.
  Since
  \[
    | E(G_{a n, q-1}) | = 2 a n^q \frac{n^q - 1}{n-1},
  \]
  and by the lower bound established in Lemma~\ref{eq:msvc_cost_bound}, the cost of the rest is
  \begin{align*}
    &\phi(G_{a n, q-1}) + 2 a n^q \frac{n^q - 1}{n-1} \cdot (2 p - a) n^q
    \\
    > & a^2 n^{3q-2} (n+2) + 2 a n^q (n^{q-1} + n^{q-2}) \cdot (2 p-a)n^q 
    \\
    = & a^2 n^{3 q-1} + 2 a^2 n^{3 q -2} + 4 a p n^{3 q-1} - 2 a^2 n^{3 q-1} + 4 a p n^{3 q-2}- 2 a^2 n^{3 q-2}.
  \end{align*}
  The total cost of~$\sigma$ is thus greater than
  \[
     p^2 n^{3q} + (p-a)^2 n^{3q-1}(n-1) + 3 p^2 n^{3q-1} + 4 a p n^{3q-2}.
  \]

  On the other hand, by Lemma~\ref{eq:msvc_cost_bound}, the solution that takes all the sets~$R_{q}^{1}\cup \cdots \cup R_{q}^{p}$ and then an optimal solution of~$G_{p n, q - 1}$ has a cost smaller than \[
    p^2 n^{3 q} + 3p^2 n^{3 q - 1} + 9p^2 n^{3 q - 2} + p n^{2 q}.
  \]
  Thus, we have
  \[
    a > \left\lfloor p - \frac{4 p}{n} \right\rfloor \ge p - \frac{5  p}{n},
  \]
  which means that only fewer than $5p/n$ copies can start with~$L$, so at least $p(n-5)/n$ copies start with~$R_{q}$.
  Therefore,
  \[
    \tauplus(G_{p, q}) 
    = (2 p - a)n^q + \tauplus\left(  G_{a n, q-1} \right) > p n^q + \tauplus\left(  G_{p (n-5), q - 1} \right).\qedhere
  \]
\end{proof}	
        
Finally, we are now ready to show the bound on $\tauplus(B_n)$.
\begin{lemma}\label{lem:msvc-bound}
  $\tauplus(B_{n}) = \Omega(n^{n+1})$. \end{lemma}
\begin{proof}
The cost of the solution starting with all vertices in~$L$ is
  \[
    \sum_{i=1}^{2n^n} i \cdot \frac{n^{n+1} - 1}{n - 1} > 2n^{3n} + 2n^{3n-1}.
  \]
  By Lemma~\ref{eq:msvc_cost_bound}, the cost of a solution starting with~$R_n$ followed by an optimal solution of~$G_{n, n-1}$ is
  \[
    \sum_{i=1}^{n^n} i \cdot 2n^n + \phi(G_{n, n-1}) < n^{3n}+ n^{2n} + n^{3n-2}(n+6) < n^{3n} + 2n^{3n-1}.
  \]

  By Lemma~\ref{eq:msvc_multiparts}, any optimal solution of~$B_{n}$ must fall into the second case, and hence
  \begin{align*}
    \tauplus(B_n) &= n^n + \tauplus(nG_{n-1})
    \\
                  &> 2n^n + \tauplus \left(  G_{n(n-5), n-2} \right)
    \\
                  &> 2n^n + n^{n-1}(n-5) + \tauplus \left( G_{n(n-5)^2, n-3} \right)
    \\
                  &> n^n + n^n \left[ 1 + \left( 1-\frac{5}{n} \right) + \left( 1-\frac{5}{n} \right)^2 + \cdots + \left( 1-\frac{5}{n} \right)^{n-3} \right] + \tauplus \left(  G_{n(n-5)^{n-2}, 1} \right).
\end{align*}
Since $n\ge 7$,
\[
  \sum_{j=0}^{n-3} \left(1-\frac{5}{n}\right)^j \ge (n-2) \left(1-\frac{5}{n}\right)^{n}
  \ge \frac{n}{2} e^{-5}.
\]
Therefore, $\tauplus(B_n) \ge c\, n^{n+1}$ for some constant $c>0$, e.g., $c = e^{-5}/2$.
  This concludes the proof.
\end{proof}

Theorem~\ref{thm:msvc} follows from Proposition~\ref{lem:vertex-cover} and Lemma~\ref{lem:msvc-bound}.
The rest of this section is devoted to the proof of Lemma~\ref{lem:main}, which we do in two steps.

\begin{lemma}\label{lem:prefix-1}
  Let~$\sigma$ be an optimal solution for~$G_{p, q}$.  For all~$i = 1, \ldots, p$, the subsequence~$\sigma^{i}$ starts with either~$R_q^i$ or a proper (possibly empty) subset of~$R_q^i$ followed by~$L^i$.
\end{lemma}
\begin{proof}
Let~$v$ be the first vertex in~$\sigma$ from~$V(B^{i}_{q})\setminus (L^i\cup R_q^i)$, and let~$u$ be the first vertex in~$\sigma$ from~$L^{i}$.
  
  In the first case, $v<_{\sigma} u$ (i.e., $v<_{\sigma} x$ for all~$x\in L^i$).
  We argue that~$x<_{\sigma} v$ for all~$x\in R_q^i$; i.e.,~$\sigma^{i}$ starts with~$R_q^i$.
  Suppose for contradiction that there exists a vertex~$w \in R_{q}^i$ such that~$\sigma(v) < \sigma(w)$.  
  Let~$\sigma'$ be the ordering of~$G_{p, q}$ obtained from~$\sigma$ by switching~$v$ and~$w$.
  Since~$N(v) \subsetneq L^i = N(w)$, we have
  \[
    r_{\sigma'}(\sigma(v)) = n^{q} > r_\sigma(\sigma(v)), 
  \]
  and $r_{\sigma'}(j) \leq r_\sigma(j)$ for all~$j$ with~$\sigma(v) < j \le \sigma(w)$.
It follows from Corollary~\ref{cor:switching} that~$\sigma'$ is strictly better than~$\sigma$, contradicting the optimality of~$\sigma$.
	
  In the second case, $u<_{\sigma} v$.  Let
  \begin{align*}
    h &= \sigma^i(u) - 1,
    \\
    k &= \sigma^i(v) - h.
  \end{align*}
  By the selection of~$u$ and~$v$, all the first~$h$ vertices in~$\sigma^i$ are from~$R_{q}^{i}$; number them as~$w_{1}, \ldots, w_{h}$ in order.  We note that the next~$k$ vertices (from~$u$, inclusive, to~$v$, exclusive) in~$\sigma^i$ are from~$L^{i}$.  If there is a vertex~$x\not\in L^{i}$ with~$u  <_\sigma x <_\sigma v$, then~$x\in R_{q}^{i}$ by the selection of~$v$, and switching~$x$ with~$u$ gives a strictly better solution~$\sigma'$ than~$\sigma$ by Corollary~\ref{cor:switching}: $r_{\sigma'}(\sigma(u)) = 2 n^{q} > r_\sigma(\sigma(u))$ and $r_{\sigma'}(j) \leq r_\sigma(j)$ for all~$j$ with~$\sigma(u) < j \le \sigma(x)$.
  It remains to show that
  \[
    h = n^{q} \quad\text{or}\quad k = 2 n^{q}.
  \]
  
  Suppose that~$h <  n^{q}$ and~$k < 2 n^{q}$.
Let~$\sigma'$ be obtained from~$\sigma$ by replacing the~$(h+1)$st through~$(h + 2 n^{q})$th in~$\sigma^{i}$ with the vertices in~$L^{i}$.
  Since~$L^{i}$ covers all edges in~$B_{q}^{i}$, the order of the remaining vertices (i.e., $R^{i}\setminus \{w_{1}, \ldots, w_{h}\}$) is immaterial.

  Number the vertices in~$L^{i}$ as~$u_{1}, \ldots, u_{2 n^{q}}$ in the order of their occurrence in~$\sigma'$, and let~$t = \sigma'(u_{2 n^{q}})$.
Note that~$\sigma'(u_{k+1}) = \sigma(v)$.
  We have
  \[
    r_{\sigma'}(\sigma'(u_{j})) = \frac{n^{q+1} - 1}{n - 1} - h
    = r_\sigma(u_{1})
    \geq r_\sigma(\sigma'(u_{j})), j = k+1, \ldots, 2 n^{q}.
  \]
  On the other hand, $r_{\sigma'}(x) = 0$ for all~$x\in R^{i}\setminus \{w_{1}, \ldots, w_{h}\}$, which come after~$u_{2 n^{q}}$ in~$\sigma'$, and hence for all~$j > t$,
  \[
    r_{\sigma'}(j) \leq r_\sigma(j).
  \]
  Let~$X$ be the set of the first~$(h + 2 n^{q})$ vertices in~$\sigma^{i}$, and~$y$ the~$(h + 2 n^{q} +1)$st vertex in~$\sigma^{i}$; note that~$\sigma(y) > t$.
  Since~$h < n^{q}$, the set~$X$ is not a vertex cover of~$B_{q}^{i}$ by Proposition~\ref{lem:vertex-cover}.  Thus, $r_\sigma(\sigma(y)) > 0 = r_{\sigma'}(\sigma(y))$, and $\sigma'$ is a strictly better solution than $\sigma$ by Proposition~\ref{lem:bettersort}, contradicting the optimality of~$\sigma$.  
                
  Therefore, either~$h = n^{q}$, i.e., the subsequence~$\sigma^{i}$ starts with~$R_q^i$, or~$h < n^{q}$ and~$k = 2 n^{q}$, i.e., it starts with a proper subset of~$R_q^i$ followed by~$L^i$.
\end{proof}
	
We are now ready for the proof of Lemma~\ref{lem:main}: in the second case of Lemma~\ref{lem:prefix-1}, the proper subset of~$R_q^i$ must be empty.

\begin{proof}[Proof of Lemma~\ref{lem:main}]
In view of Lemma~\ref{lem:prefix-1}, we need to show that for all~$i = 1, \ldots, p$, the subset of~$R_q^i$ must be empty in the second case.
  Our strategy is to assume that some copy starts with a nonempty proper prefix of~$R_q^i$, and then locally reorder the boundary between this prefix and the subsequent block of~$L^{i}$-vertices.  By comparing the change in edge costs under these swaps/rotations, we show that one obtains a strictly better solution, contradicting optimality.

  We number the vertices in~$L^{i}$ as~$u_{1}^{i}, \ldots, u_{2 n^{q}}^{i}$, and the vertices in~$R_{q}^{i}$ as~$w_{1}^{i}, \ldots, w_{n^{q}}^{i}$, by their occurrence in~$\sigma$; let~$h_{i}$ denote the length of the longest prefix that is a subset of~$R_{q}^{i}$, and~$\ell^{i}$ the number of consecutive vertices from~$L^{i}$ afterward.

  Suppose for contradiction that~$\sigma^{1}$ starts with a proper and nonempty subset of~$R_q^1$ followed by~$L^1$ (we can always renumber the copies of~$B_{q}$).

  We show that at most one copy can behave in this way; otherwise we can locally swap and improve the solution.
  In the first case, there exists another~$b\in \{1, \ldots, p\}$ such that~$\sigma^{b}$ starts with a proper and nonempty subset of~$R_q^b$ followed by~$L^b$.  Again, assume without loss of generality that~$b = 2$ and~$h_{1} \le h_{2}$.
Hence, $r_\sigma(u_{i}^{1}) \geq r_\sigma(u_{j}^{2})$ for all $i, j \in \{1, \dots, 2n^q\}$.
  Since~$r_\sigma(w_{i}^{1}) \geq r_\sigma(w_{j}^{2}) = 2 n^{q}$ for all $i, j \in \{1, \dots, n^q\}$, we may assume that the following vertices in this order in $\sigma$:
  \[
    w_{1}^{1}, \ldots, w_{h_{1}}^{1},
    w_{1}^{2}, \ldots, w_{h_{2}}^{2},
    u_{1}^{1}, \ldots, u_{2 n^{q}}^{1},
    u_{1}^{2}, \ldots, u_{2 n^{q}}^{2}.
  \]
  Let~$\sigma'$ be obtained from~$\sigma$ by swapping $w_{h_{1}}^{1}$ and $w_{h_{2}+1}^{1}$.  Note that~$r_{\sigma'}(w_{h_{2}+1}^{1}) = r_{\sigma}(w_{h_{1}}^{1}) = 2 n^{q}$ and~$r_{\sigma}(w_{h_{2}+1}^{1}) = r_{\sigma'}(w_{h_{1}}^{1}) = 0$.
For $i = 1, \dots, 2n^{q}$, we have
  \begin{align*}
    r_{\sigma'}(u_{i}^{1}) = r_{\sigma}(u_{i}^{1}) + 1 > r_\sigma(u_{i}^{1}),
    \\
    r_{\sigma'}(u_{i}^{2}) = r_{\sigma}(u_{i}^{2}) - 1 < r_\sigma(u_{i}^{2}).
  \end{align*}
    Let $t = \sigma'(u_{2 n^{q}}^{1}) = \sigma(u_{2 n^{q}}^{1})$.
    Finally,~$r_{\sigma}(j) = r_{\sigma'}(j) = 0$ for all~$j >\sigma(u_{2 n^{q}}^{2}) = t + 2 n^{q}$.
    Thus, the solution~$\sigma'$ is a strictly better solution than $\sigma$ by Proposition~\ref{lem:bettersort}, contradicting the optimality of~$\sigma$.  

    In the rest, for all~$i = 2, \ldots, p$, the subsequence~$\sigma^{i}$ starts with either~$R_{q}^{i}$ or~$L^{i}$.
    We may renumber the copies such that~$\sigma^{i}$ starts with~$R_{q}^{i}$ when~$i = 2, \ldots, a$ and~$L^{i}$ when~$i = a+1, \ldots, p$.
  Hence, $\sigma$ must start with the following in order
  \begin{itemize}
  \item all vertices in~$R_{q}^2 \cup R_{q}^3 \cup \cdots \cup R_{q}^{a} \cup \{ w_{1}^{1}, \ldots, w_{h_{1}}^{1} \}$, each vertex in it with effective coverage~$2 n^{q}$;
  \item all vertices in $L^{a+1}, L^{a+2}, \dots, L^{p}$, each vertex in it with effective coverage~$\frac{n^{q+1} - 1}{n - 1}$; and
  \item all vertices in $L^{1}$, each vertex in it with effective coverage~$\frac{n^{q+1} - 1}{n - 1} - h_{1}$.
  \end{itemize}
We may assume that $w_{h_{1}}^{1}$ is the last vertex in~$\sigma$ with effective coverage~$2 n^{q}$ and~$w_{h_{1}+1}^{1}$ is the last of~$\sigma$, i.e., 
  \begin{align*}
    \sigma(w_{h_{1}}^{1}) =& (a - 1) n^{q} + h_{1},
    \\
    \sigma(w_{h_{1}+1}^{1}) =& p(q+4)n^q.
  \end{align*}
  Thus,~$\sigma$ has the following structure:
  \[
     \underbrace{\big[ R_{q}^2, \cdots, R_{q}^{a} \big]}_{\text{full $R_{q}$-blocks}},
     \underbrace{\big[ w_{1}^{1}, \cdots, w_{h_{1}}^{1} \big]}_{\text{partial prefix}},
     \underbrace{\big[  L^{a+1}, L^{a+2}, \dots, L^{p} \big]}_{\text{full $L$-blocks}},
     \underbrace{\big[ u_{1}^{1}, \cdots, u_{2 n^{q}}^{1} \big]}_{L^{1}}, \quad\cdots,\quad w_{h_{1}+1}^{1}.
  \]

  We now construct two new solutions $\overrightarrow{\sigma}$ and $\overleftarrow{\sigma}$ by rotating the subsequence between them (both inclusive) of $\sigma$ by one position.
  Define~$\overleftarrow{\sigma}$ by moving~$w_{h_{1}}^{1}$ to the very end of the sequence, shifting all vertices after~$(a - 1) n^{q} + h_{1}$ one position to the left, and~$\overrightarrow{\sigma}$ by moving~$w_{h_{1}+1}^{1}$ to the position $(a - 1) n^{q} + h_{1} + 1$, shifting all other vertices after~$(a - 1) n^{q} + h_{1}$ one position to the right.
  We show that at least one of~$\overrightarrow{\sigma}$ or~$\overleftarrow{\sigma}$ is a better solution than~$\sigma$.

  In the transformation from~$\sigma$ to~$\overrightarrow{\sigma}$, the costs of the edges between~$w_{h_{1}+1}^{1}$ and~$L^{1}$ become $\overrightarrow{\sigma}(w_{h_{1}+1}^{1})= (a - 1) n^{q} + h_{1} + 1$, while every edge that was originally covered after position~$(a - 1) n^{q} + h_{1}$ and is not incident to~$w_{h_{1}+1}^{1}$ has its cost increased by~1.
  Thus, the cost increases by
  \begin{align*}
    &\sum_{i = (a - 1) n^{q} + h_{1} + 1}^{|V(G_{p, q})|} \left( r_\sigma(i) - r_{\overrightarrow{\sigma}}(w_{h_{1}+1}^{1})  \right) - \sum_{i = 1}^{2n^{q}} \left( \sigma(u_i^{1}) - \overrightarrow{\sigma}(w_{h_{1}+1}^{1}) \right)
    \\
    =& \left( |E(G_{p, q})| - 2n^q ((a-1) n^q + h_{1}) - 2n^q \right) - \left( \sum_{i = 1}^{2 n^q} \left( \sigma(u_i^{1}) - \sigma(w_{h_{1}}^{1}) \right) - 2n^q \right).
    \\
    =& |E(G_{p, q})| - 2n^{q} ((a-1) n^q + h_{1}) - \sum_{i = 1}^{2n^b} \left( \sigma(u_i^{1}) - \sigma(w_{h_{1}}^{1}) \right).
  \end{align*}
	
  In the transformation from~$\sigma$ to~$\overleftarrow{\sigma}$, the cost of the edge between~$w_{h_{1}+1}^{1}$ and~$u_{i}^{1}, i= 1, \ldots, 2 n^{q}$, becomes $\overleftarrow{\sigma}(u_{i}^{1})$, while every edge that was originally covered after position~$(a - 1) n^{q} + h_{1}$ has its cost decreased by~1.
  Thus, the cost increases by
  \begin{align*}
    &\sum_{i = 1}^{2n^q} \left( \overleftarrow{\sigma}(u_i^{1}) - \sigma(w_{h_{1}}^{1}) \right) - \sum_{i = \sigma(w_{h_{1}}^{1}) + 1}^{|V(G_{p, q})|} r_\sigma(i) 
    \\
    =& \sum_{i = 1}^{2n^q} \left( \sigma(u_i^{1}) - 1 - \sigma(w_{h_{1}}^{1}) \right) - \left( |E(G_{p, q})| - 2n^q ((a-1) n^q + h_{1}) \right).
  \end{align*}

  Since the summation of these two terms is~$-2n^q$, at least one of them is negative.  Hence, at least one of~$\overrightarrow{\sigma}$ and~$\overleftarrow{\sigma}$ is strictly better than~$\sigma$, contradicting the optimality of~$\sigma$.
  This concludes the proof. 
\end{proof}

\section{A parameterized algorithm}

We present a simple fixed-parameter branching algorithm for the \mssc{} problem on bounded-rank hypergraphs.  
Recall that the rank~$r(H)$ of a hypergraph~$H$ is the maximum cardinality of any of the hyperedges in~$H$, and we consider the decision version: given a hypergraph $H$ and integers $k$ and $w$, does there exist a set cover~$S$ of~$H$ and a permutation of~$S$ such that~$|S| \le k$ and the cost is at most~$w$.

We say that a solution~$\sigma$ of an instance~$(H, w, k)$ is \emph{feasible} if~$|S_\sigma| \le k$ and its total cost is at most~$w$.
For the minimum sum vertex cover problem, the set of any feasible solution must contain all vertices with degree greater than~$k$.  This is no longer true in our setting. Instead, we require the Sunflower Lemma~\cite{erdos-60-sunflower}.

A \emph{sunflower} in a hypergraph~$H$ consists of a (possibly empty) subset~$K \subseteq V(H)$, called the \emph{core}, and a subset~$P \subseteq E(H)$, called the \emph{petals}, such that the intersection of every pair of distinct hyperedges in~$P$ is exactly~$K$.  The size of the sunflower is~$|P|$.

\begin{lemma}[Sunflower Lemma~\cite{erdos-60-sunflower}]
  If a hypergraph~$H$ with no repeated edges has more than~$r(H) \cdot r(H)! (k - 1)^{r(H)}$ hyperedges, then $H$ contains a sunflower with $k$ petals, and such a sunflower can be computed in time polynomial in $|E(H)|$, $|V(H)|$, and $k$.
\end{lemma}

If there is a sunflower with $k+1$ petals, then any set cover of~$H$ with at most~$k$ vertices must contain at least one vertex from its core~$K$.
Since~$K$ cannot be larger than the size of any petal, we have~$|K| \le r(H)$.  Thus, we can branch on choosing a vertex from~$K$: this is why our algorithm only works on bounded-rank hypergraphs.
We repeat this until there is no sunflower with $k+1$ petals, or~$k$ vertices have already been selected.
By the Sunflower Lemma, the number of edges, and hence the number of vertices, in the remaining instance is bounded.  We can then find the remaining vertices of the solution set by branching, and finally enumerate all possible orderings of this set.

We now summarize the algorithm in Figure~\ref{alg:main} and use it to prove Theorem~\ref{thm:alg}.

\begin{figure}[h!]
  \centering 
  \begin{tikzpicture}
    \path (0,0) node[text width=.85\textwidth, inner xsep=20pt, inner ysep=10pt] (a) {
      \begin{minipage}[t!]{\textwidth}
        \begin{tabbing}
          AAA\=Aaa\=aaa\=Aaa\=MMMMMAAAAAAAAAAAA\=A \kill
          1. \> $S\leftarrow \emptyset$;
          \\
          2. \> \textbf{while} $|E(H - S)| > r \cdot r!(k - |S|)^r$ and $|S| \leq k$ \textbf{do}
          \\
          2.1. \>\> \textbf{if} $|S| = k$ \textbf{then return} ``no'';
          \\
          2.2. \>\> find a sunflower of $H - S$ with core~$K$;
          \\
          2.3. \>\> \textbf{if} $K = \emptyset$ \textbf{then return} ``no'';
          \\
          2.4. \>\> guess a vertex $v \in K$ and add it to~$S$;
          \\
          3. \> remove all isolated vertices in~$H - S$;
          \\
          4. \> guess a subset $T\subseteq V(H - S)$ with $|T| \le k - |S|$ and $S\leftarrow S\cup T$;
          \\
          5. \> guess an bijective mapping $\sigma \colon S \rightarrow \{1,2,\dots, |S|\}$;
          \\
          6. \> \textbf{if} $S$ is a set cover of~$H$ and the cost of~$\sigma$ is~$\le w$ \textbf{then return} ``yes'';
          \\
          \> {\bf else return} ``no.''
        \end{tabbing}
      \end{minipage}
    };
    \draw[draw=gray!60] (a.north west) -- (a.north east) (a.south west) -- (a.south east);
  \end{tikzpicture}
  \caption{The parameterized algorithm for \mssc{}.}
  \label{alg:main}
\end{figure}

\begin{proof}[Proof of Theorem~\ref{thm:alg}]
  We use the algorithm described in Figure~\ref{alg:main}.
  Since the algorithm returns ``yes'' only when a solution is explicitly constructed, it suffices to show that a solution must be returned for every yes-instance.
  Let $S$ denote the set of elements contained in an optimal solution, and let~$r = r(H)$.
  
  The algorithm has two phases: first, finding the vertex set~$S$ (steps~1--4), and then computing an ordering of~$S$ (step~5).
  We find the solution set~$S$ by branching.
  During this process, we work with the hypergraph~$H - S$: all hyperedges intersecting~$S$ have already been covered. 
  If~$H - S$ still has more than~$r \cdot r!(k - |S|)^r$ hyperedges, then by the Sunflower Lemma we can find a sunflower with $k - |S| + 1$ petals~\cite{erdos-60-sunflower}.
  Any solution must contain a vertex from the core of this sunflower; otherwise, it would need to use a distinct vertex to cover each petal, which would require more than~$k$ vertices and thus be infeasible.
  Thus, there cannot be a feasible solution when the core is empty.
  Therefore, either step~2.3 correctly returns ``no,'' or at least one branch in step~2.4 adds a vertex from the optimal solution.

  Once the algorithm exits the loop in step~2, we are allowed to add at most~$k - |S|$ additional vertices, and these can be found by enumeration: we branch over all subsets $T\subseteq V(H - S)$ with $|T|\le k - |S|$.

    We now analyze the running time.
  The search tree in step~2 has depth at most~$k$, and the branching factor is at most the core size, which is bounded by~$r$. Thus, this phase takes time~$O(r^{|S|+1} \cdot r!(k - 1)^r m)$, where $m = |E(H)|$.
  After step~2, we have~$|E(H - S)| \le r \cdot r!(k - |S|)^r$ by the Sunflower Lemma, and hence~$|V(H - S)| \le r |E(H - S)| \le r^2 \cdot r!(k - |S|)^r$ because there are no isolated vertices. Thus, there are~$O\bigl((r^2 \cdot r! (k-1)^r)^{k-|S|}\bigr)$
  subsets to be enumerated in step~4.
  Step~5 takes~$O(k!)$ time. 
Overall, the total running time is~$O\left( r^{r+k} k^r m + (rk)^{rk} \right)$. 
\end{proof}
Let us remark that we have not made an attempt to optimize the running time.  For example, Step~5 can be improved using the Held--Karp dynamic program~\cite{held-62-dp}.

\section{Concluding remarks}\label{sec:conclusion}

The most compelling open question is to close the $\log\log \tau(G)$ gap between the upper bound in Theorem~\ref{thm:mssc-rank} for graphs and the lower bound  given by the construction in Theorem~\ref{thm:msvc}.

We also leave open the fixed-parameter tractability of minimum sum set cover on general hypergraphs when parameterized solely by~$\tauplus$.
Since~$\tauplus$ and~$\tau$ are polynomially related (Theorem~\ref{thm:mssc-rank}), minimum sum vertex cover is fixed-parameter tractable with respect to~$\tau$ \cite{cao-25-msvc}.  For hypergraphs, this is very unlikely to remain true.
In fact, we conjecture that minimum sum set cover is NP-hard even on hypergraphs with bounded minimum set cover number.
One possible approach toward proving this would be to adapt our reduction from Theorem~\ref{thm:main-1}, for which
\[
  \phi(H_{G}) = 3 \phi(G) + \sum_{i = 1}^n 3 \left( 2^{n - i} - (n - i) \right).
\]
However, this reduction produces instances whose size grows exponentially.

\end{document}